\documentclass[a4paper,10pt]{amsart}

\usepackage[colorlinks,linkcolor=blue,citecolor=blue]{hyperref}
\usepackage{latexsym, amssymb, amsmath, amscd, amsthm, mathrsfs, bbm}
\usepackage[all, knot]{xy}
\xyoption{arc}

\usepackage{anysize}\marginsize{20mm}{20mm}{25mm}{25mm}
\newcommand{\bea}{\begin{eqnarray}}
\newcommand{\eea}{\end{eqnarray}}
\def\be{\begin{equation}}
\def\ee{\end{equation}}

\newcommand{\ie}{{\it i.e.~}}

\newcommand{\bpm}{\begin{pmatrix}}
\newcommand{\epm}{\end{pmatrix}}
\newcommand{\bmm}{\begin{matrix}}
\newcommand{\emm}{\end{matrix}}

\newcommand {\emptycomment}[1]{}

\usepackage{centernot}
\newcommand{\nn}{{\nonumber}}
\newcommand{\Sec}[1]{Sec.~\ref{#1}}

\usepackage{enumitem} 
\usepackage{mathtools,amssymb,varwidth}

\newcommand{\Fig}[1]{Fig.~\ref{#1}}

\usepackage{datetime}

\newtheorem{theorem}{Theorem}[section]

\newtheorem{lemma}[theorem]{Lemma}
\newtheorem{assumption}[theorem]{Assumption}

\newtheorem{proposition}[theorem]{Proposition}

\newtheorem{convention}[theorem]{Convention}

\newcommand{\Gr}{\mathrm{Gr}}

\newcommand{\ceil}[1]{\left\lceil #1 \right\rceil}
\newcommand{\floor}[1]{\left\lfloor #1 \right\rfloor}

\begin{document}
\sloppy
\title{A matrix solution to any polygon equation}

\author{Zheyan Wan}

\address{Beijing Institute of Mathematical Sciences and Applications, Beijing 101408, China \\
wanzheyan@bimsa.cn}

\maketitle

\begin{abstract}
In this article, we construct matrices associated to Pachner $\frac{n-1}{2}$-$\frac{n-1}{2}$ moves for odd $n$ and matrices associated to Pachner $(\frac{n}{2}-1)$-$\frac{n}{2}$ moves for even $n$. The entries of these matrices are rational functions of formal variables in a field. We prove that these matrices satisfy the $n$-gon equation for any $n$.

\end{abstract}

\section{Introduction}

In this article, an $n$-gon means a simplicial complex with $n$ vertices. Pachner moves \cite{Pachner,Lickorish} are operations on a simplicial complex that changes its triangulation to another one. For example, a Pachner 2-2 move is a flip of diagonals in a quadrilateral, a Pachner 2-3 move is an operation that replaces two tetrahedra 1245, 2345 with three tetrahedra 1234, 1235, 1345, see \Fig{fig:Pachner}, and so on.
\begin{figure}[h]
\centering\includegraphics[width = 0.9\textwidth]{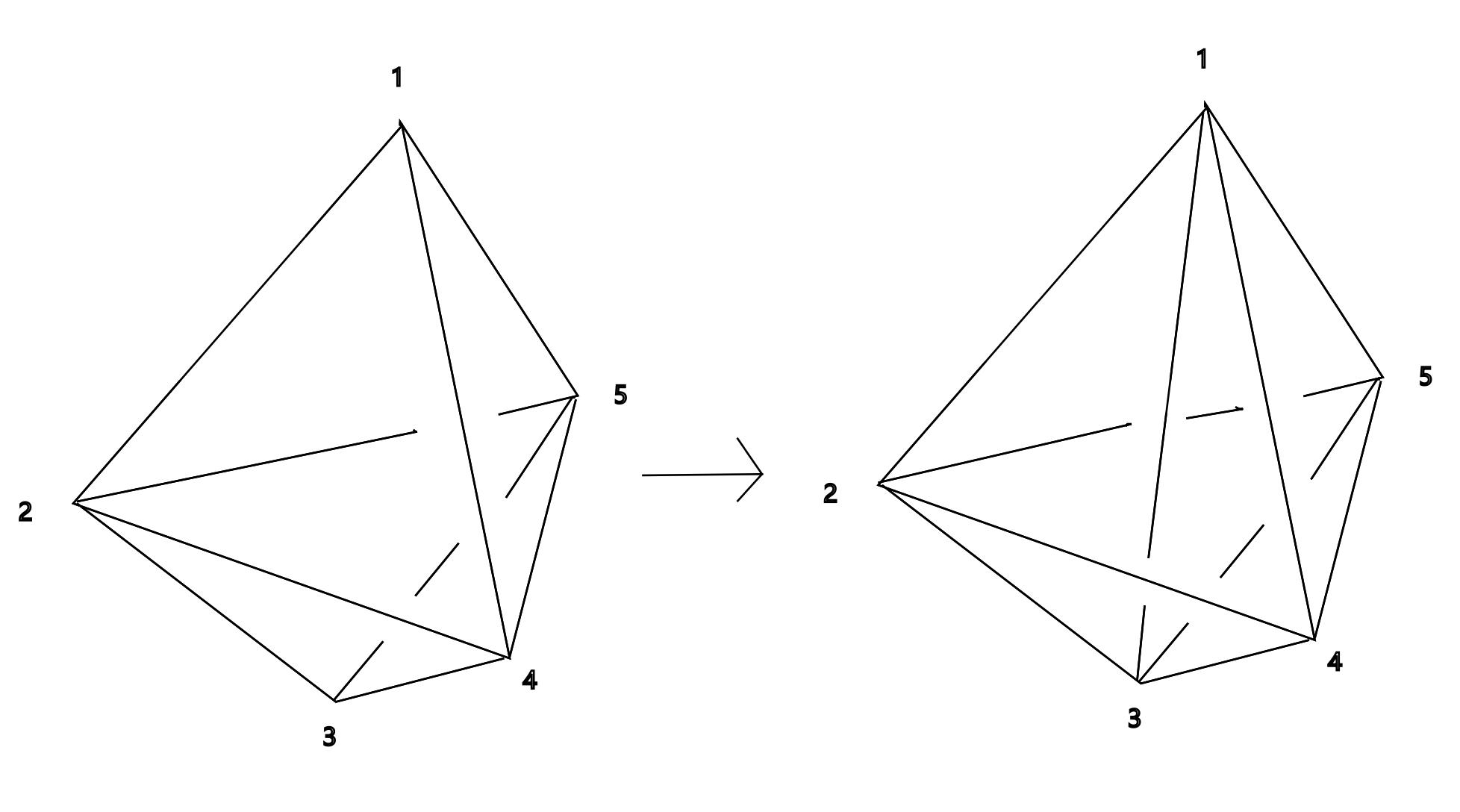}
\caption{Pachner 2-3 move.}\label{fig:Pachner}
\end{figure}

The polygon equation is closely related to the Pachner moves. 
More precisely, the $n$-gon equation is an equation between two products of Pachner moves. If $n$ is odd, then the left-hand side of the $n$-gon equation is the product of $\frac{n-1}{2}$ Pachner $\frac{n-1}{2}$-$\frac{n-1}{2}$ moves, while the right-hand side of the $n$-gon equation is the product of $\frac{n+1}{2}$ Pachner $\frac{n-1}{2}$-$\frac{n-1}{2}$ moves.
If $n$ is even, then the left-hand side of the $n$-gon equation is the product of $\frac{n}{2}$ Pachner $(\frac{n}{2}-1)$-$\frac{n}{2}$ moves, while the right-hand side of the $n$-gon equation is the product of $\frac{n}{2}$ Pachner $(\frac{n}{2}-1)$-$\frac{n}{2}$ moves. For example, the pentagon equation is an equation between the product of 2 Pachner 2-2 moves and the product of 3 Pachner 2-2 moves, the hexagon equation is an equation between the product of 3 Pachner 2-3 moves and the product of 3 Pachner 2-3 moves, and so on.

The polygon equation itself can also be regarded as a Pachner move. More precisely, if $n$ is odd, then the $n$-gon equation can be regarded as a $\frac{n-1}{2}$-$\frac{n+1}{2}$ move (that replaces the left-hand side with the right-hand side, both sides consist of a sequence of Pachner $\frac{n-1}{2}$-$\frac{n-1}{2}$ moves and each Pachner $\frac{n-1}{2}$-$\frac{n-1}{2}$ move in both sides corresponds to a face of the $n$-gon (regarded as a simplicial complex)). If $n$ is even, then the $n$-gon equation can be regarded as a $\frac{n}{2}$-$\frac{n}{2}$ move (that replaces the left-hand side with the right-hand side, both sides consist of a sequence of Pachner $(\frac{n}{2}-1)$-$\frac{n}{2}$ moves and each Pachner $(\frac{n}{2}-1)$-$\frac{n}{2}$ move in both sides corresponds to a face of the $n$-gon (regarded as a simplicial complex)).

In \cite{DK21}, a matrix solution for any odd polygon (odd-gon) equation was constructed. Namely, for odd $n$, a $\frac{n-1}{2}\times\frac{n-1}{2}$ matrix (whose entries are rational functions of formal variables in a field) was constructed for each Pachner $\frac{n-1}{2}$-$\frac{n-1}{2}$ move, and it was proved in \cite{DK21} that these matrices satisfy the odd-gon equation. In the present article, we do similar things for any even polygon (even-gon) equation. Namely, for even $n$, we construct a $\frac{n}{2}\times(\frac{n}{2}-1)$ matrix (whose entries are rational functions of formal variables in a field) for each Pachner $(\frac{n}{2}-1)$-$\frac{n}{2}$ move, and we prove that these matrices satisfy the even-gon equation. Our method also applies to any odd-gon equation and our solution for any odd-gon equation is exactly a partial case of the one in \cite{DK21}.
A similar but different solution for any odd-gon equation was given in \cite{K22b}. 
A different solution to the hexagon equation was given in \cite{KS17}.
These solutions to the polygon equations are useful to study polygon relations appearing from simplicial cocycles and invariants of manifolds, like in \cite{K21,K23}.

This article is structured as follows: In \Sec{sec:polygon}, we formulate the polygon equation.
In \Sec{sec:main}, we formulate our main results (a matrix solution for any polygon equation). In \Sec{sec:proof}, we prove our main results. 
In \Sec{sec:pentagon}, we give an example: solving the pentagon equation. In \Sec{sec:hexagon}, we give an example: solving the hexagon equation. In \Sec{sec:heptagon}, we give an example: solving the heptagon equation.

\subsection{Acknowledgements}
The author is grateful to I.G. Korepanov for the helpful discussion and comments.

\section{The polygon equation}\label{sec:polygon}
In this section, we formulate the polygon equation.

The $n$-gon equation involves two sequences of Pachner moves. Each sequence is a path from the initial triangulation to the final triangulation of the $n$-gon. The number of $(n-3)$-simplices in each step of the two paths does not change for odd $n$ and increases by 1 for even $n$. Each $(n-3)$-simplex has $n-2$ vertices, it corresponds to a pair $(i,j)$ where $1\le i<j\le n$ and $i,j$ are not the vertices of the $(n-3)$-simplex.

The initial triangulation of the $n$-gon consists of $(n-3)$-simplices which correspond to the pairs
\bea\label{eq:pairs}
(i,j)=(n+1-2k,n+2-2l)\text{ for }1\le l\le k\le \floor{\frac{n-1}{2}}
\eea
where $\floor{x}$ means the largest integer less than or equal to $x$. There are $\dfrac{\floor{\frac{n-1}{2}}(\floor{\frac{n-1}{2}}+1)}{2}$ $(n-3)$-simplices in the initial triangulation of the $n$-gon. For example, the initial triangulation of the pentagon consists of triangles 123, 134, and 145 which correspond to the pairs (4,5), (2,5), and (2,3). The initial triangulation of the hexagon consists of tetrahedra 1234, 1245, and 1256 which correspond to the pairs (5,6), (3,6), and (3,4). See \Sec{sec:pentagon} and \Sec{sec:hexagon}.

We can do two Pachner moves for the initial triangulation of the $n$-gon, leading to two sequences of Pachner moves on the left-hand side and right-hand side of the $n$-gon equation.
Note that there are $\floor{\frac{n-1}{2}}$ pairs in \eqref{eq:pairs} with the smallest $i$ and $\floor{\frac{n-1}{2}}$ pairs in \eqref{eq:pairs} with the largest $j$. The two Pachner moves for the initial triangulation of the $n$-gon act on the $(n-3)$-simplices corresponding to the $\floor{\frac{n-1}{2}}$ pairs in \eqref{eq:pairs} with the smallest $i$ and $\floor{\frac{n-1}{2}}$ pairs in \eqref{eq:pairs} with the largest $j$ respectively.
We denote these two sets of pairs by 
\bea
(i_{\min}, j_1),(i_{\min}, j_2),\dots,(i_{\min}, j_{\floor{\frac{n-1}{2}}})
\eea
and
\bea
(i_1,j_{\max}),(i_2,j_{\max}),\dots,(i_{\floor{\frac{n-1}{2}}},j_{\max})
\eea
respectively.
Each Pachner move in the $n$-gon equation only involves $n-1$ vertices, we denote the remaining vertex by $q$. Let $\{1,2,\dots,n\}\setminus\{q\}=\{a_1,a_2,\dots,a_{n-1}\}$. 
We divide the set $\{a_1,a_2,\dots,a_{n-1}\}$ into two parts:
\bea
\{a_1,a_2,\dots,a_{n-1}\}=\{b_1,b_2,\dots,b_{\floor{\frac{n-1}{2}}}\}\cup\{c_1,c_2,\dots,c_{\ceil{\frac{n-1}{2}}}\}
\eea
where $\ceil{x}$ means the smallest integer greater than or equal to $x$.

We denote the Pachner move that replaces the $(n-3)$-simplices corresponding to the pairs $(b_1,q)$, $(b_2,q)$, $\dots$, $(b_{\floor{\frac{n-1}{2}}},q)$ with the $(n-3)$-simplices corresponding to the pairs $(c_1,q)$, $(c_2,q)$, $\dots$, $(c_{\ceil{\frac{n-1}{2}}},q)$ by $d_{b_1b_2\cdots b_{\floor{\frac{n-1}{2}}},c_1c_2\cdots c_{\ceil{\frac{n-1}{2}}}}$ or $d_{b_1b_2\cdots b_{\floor{\frac{n-1}{2}}}}^{(q)}$ for short.
Then the two Pachner moves for the initial triangulation of the $n$-gon are $d_{j_1j_2\cdots j_{\floor{\frac{n-1}{2}}}}^{(i_{\min})}$ and $d_{i_1i_2\cdots i_{\floor{\frac{n-1}{2}}}}^{(j_{\max})}$ respectively. 

\begin{convention}
In this article, we list the elements of all sets always in increasing order, if not explicitly stated otherwise.
\end{convention}

For odd $n$, $i_{\min}=2$ and $\{j_1,j_2,\dots,j_{\frac{n-1}{2}}\}=\{3,5,\dots,n\}$, $j_{\max}=n$ and $\{i_1,i_2,\dots,i_{\frac{n-1}{2}}\}=\{2,4,\dots,n-1\}$.
The $n$-gon equation is 
\bea\label{eq:n-gon-odd}
d_{35\cdots n}^{(2)}d_{257\cdots n}^{(4)}\cdots d_{24\cdots(n-3)n}^{(n-1)}=d_{24\cdots(n-1)}^{(n)}d_{24\cdots(n-3)n}^{(n-2)}\cdots d_{35\cdots n}^{(1)}.
\eea
For even $n$, $i_{\min}=3$ and $\{j_1,j_2,\dots,j_{\frac{n}{2}-1}\}=\{4,6,\dots,n\}$, $j_{\max}=n$ and $\{i_1,i_2,\dots,i_{\frac{n}{2}-1}\}=\{3,5,\dots,n-1\}$.
The $n$-gon equation is 
\bea\label{eq:n-gon-even}
d_{46\cdots n}^{(3)}d_{368\cdots n}^{(5)}\cdots d_{35\cdots(n-3)n}^{(n-1)}d_{35\cdots(n-1)}^{(1)}=d_{35\cdots(n-1)}^{(n)}d_{35\cdots(n-3)n}^{(n-2)}\cdots d_{46\cdots n}^{(2)}.
\eea

For odd $n$, the final triangulation of the $n$-gon consists of $(n-3)$-simplices which correspond to the pairs
\bea
(i,j)=(n-2k,n+1-2l)\text{ for }1\le l\le k\le {\frac{n-1}{2}}.
\eea
There are $\dfrac{{(\frac{n-1}{2})}({\frac{n+1}{2}})}{2}$ $(n-3)$-simplices in the final triangulation of the $n$-gon.
For even $n$, the final triangulation of the $n$-gon consists of $(n-3)$-simplices which correspond to the pairs
\bea
(i,j)=(n-2k,n+1-2l)\text{ for }1\le l\le k\le {\frac{n}{2}}-1\text{ and }(1,n+2-2l)\text{ for }1\le l\le \frac{n}{2}. 
\eea
There are $\dfrac{{\frac{n}{2}}({\frac{n}{2}+1})}{2}$ $(n-3)$-simplices in the initial triangulation of the $n$-gon.
For example, the final triangulation of the pentagon consists of triangles 125, 235, and 345 which correspond to the pairs (3,4), (1,4), and (1,2). The final triangulation of the hexagon consists of tetrahedra 1236, 1346, 1456, 2345, 2356, and 3456 which correspond to the pairs (4,5), (2,5), (2,3), (1,6), (1,4), and (1,2). See \Sec{sec:pentagon} and \Sec{sec:hexagon}.

The $n$-gon equation is shown in \Fig{fig:n-gon}.
\begin{figure}[h]
\centering\includegraphics[width = 0.9\textwidth]{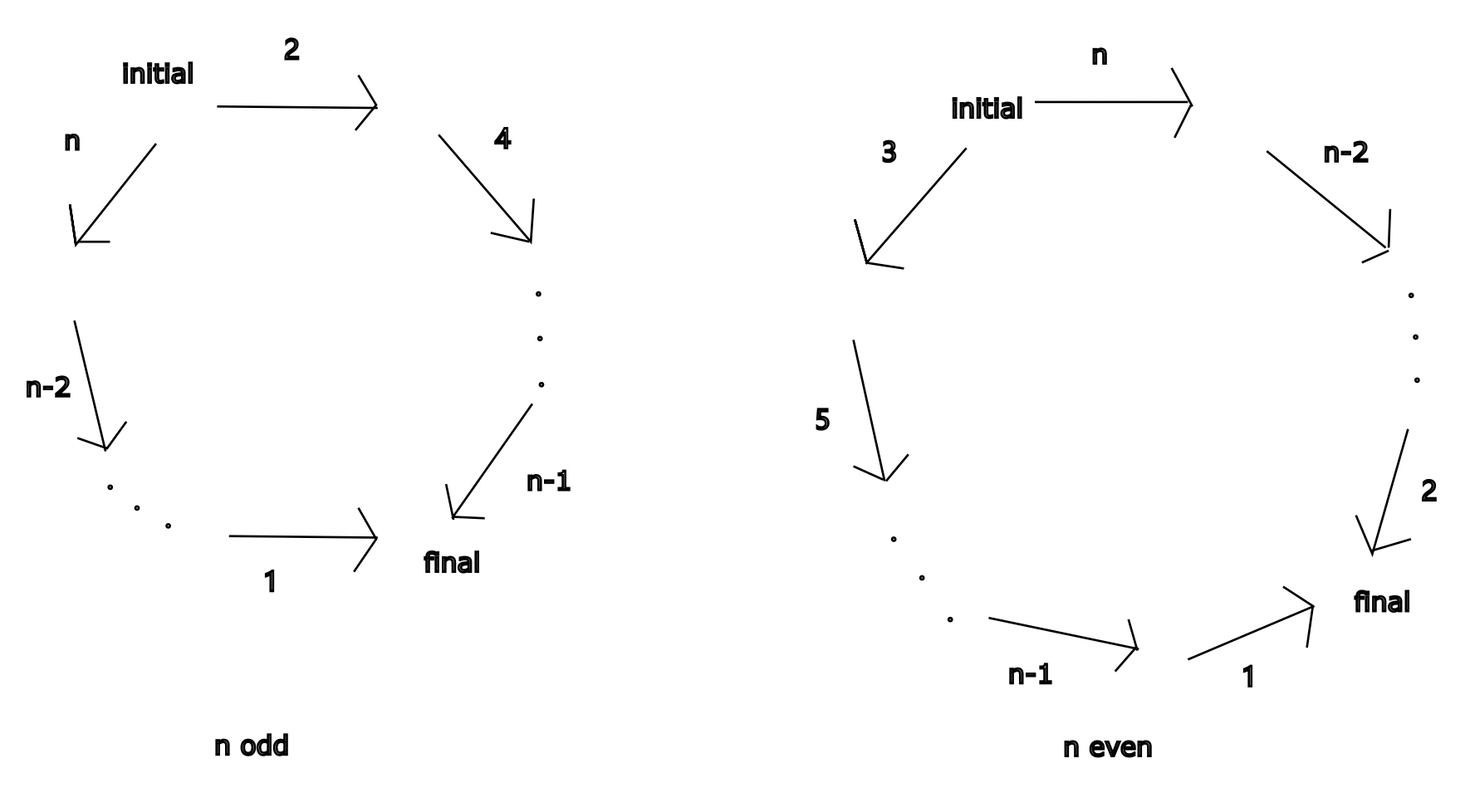}
\caption{The $n$-gon equation involves two sequences of Pachner moves from the initial triangulation to the final triangulation of the $n$-gon. Each arrow with label $q$ denotes the Pachner move which does not involve the vertex $q$.}\label{fig:n-gon}
\end{figure}

\section{Main results}\label{sec:main}

In this section, we formulate our main results (a matrix solution for any polygon equation), and we will give some examples later. 

\subsection{Construction of matrices associated to Pachner moves}\label{sec:matrices}

 For any integer $n\ge5$, let $\zeta_i$ ($i=1,2,\dots,n$) be pairwise distinct variables in a field $F$. We need to assume that the characteristic of $F$ is zero to prove our main results.
For $1\le k\le n$ and $1\le i_1<i_2<\cdots<i_{k}\le n$, we denote
\bea\label{eq:Vandermonde}
V(i_1,i_2,\dots,i_{k})=\begin{vmatrix}1 &1 & \cdots &1\\
\zeta_{i_1}&\zeta_{i_2}&\cdots&\zeta_{i_{k}}\\
\zeta_{i_1}^2&\zeta_{i_2}^2&\cdots&\zeta_{i_{k}}^2\\
\vdots&\vdots&\ddots&\vdots\\
\zeta_{i_1}^{k-1}&\zeta_{i_2}^{k-1}&\cdots&\zeta_{i_{k}}^{k-1}\\
\end{vmatrix} 
\eea
as the Vandermonde determinant.

If $n$ is odd, we associate a $\frac{n-1}{2}\times\frac{n-1}{2}$ matrix $P_{a_2a_4a_6\cdots a_{n-1},a_1a_3a_5\cdots a_{n-2}}$ ($1\le a_1<a_3<\cdots<a_{n-2}\le n$, $1\le a_2<a_4<\cdots<a_{n-1}\le n$, $\{a_1,a_3,\dots,a_{n-2}\}\cap \{a_2,a_4,\dots,a_{n-1}\}=\emptyset$) to the Pachner $\frac{n-1}{2}$-$\frac{n-1}{2}$ move for the $(n-1)$-gon $a_1a_2a_3\cdots a_{n-1}$ (regarded as a simplicial complex) that replaces the $(n-3)$-simplices $a_1a_2a_3\cdots a_{n-2}$, 
$a_1a_2a_3\cdots a_{n-4}a_{n-2}a_{n-1}$, 
$\dots$, 
$a_1a_3a_4\cdots a_{n-2}a_{n-1}$ 
with the $(n-3)$-simplices $a_1a_2a_3\cdots a_{n-3}a_{n-1}$,
 $a_1a_2a_3\cdots a_{n-5}a_{n-3}a_{n-2}a_{n-1}$, $\dots$, 
 $a_2a_3a_4\cdots a_{n-1}$ as follows: the $(i,j)$-entry of the matrix $P_{a_2a_4a_6\cdots a_{n-1},a_1a_3a_5\cdots a_{n-2}}$ is
\bea\label{eq:P-odd}
(P_{a_2a_4a_6\cdots a_{n-1},a_1a_3a_5\cdots a_{n-2}})_i^j=(-1)^{j+\frac{n-1}{2}}\frac{V(a_{n-2i},a_2,a_4,\dots,\hat{a}_{n+1-2j},\dots,a_{n-1})}{V(a_2,a_4,a_6,\dots,a_{n-1})}
\eea
where $\hat{a}_k$ means that $a_k$ is omitted.

If $n$ is even, we associate a $\frac{n}{2}\times(\frac{n}{2}-1)$ matrix $P_{a_1a_3a_5\cdots a_{n-3},a_2a_4a_6\cdots a_{n-2}a_{n-1}}$ ($1\le a_1<a_3<\cdots<a_{n-3}\le n$, $1\le a_2<a_4<\cdots<a_{n-2}<a_{n-1}\le n$, $\{a_1,a_3,\dots,a_{n-3}\}\cap \{a_2,a_4,\dots,a_{n-2},a_{n-1}\}=\emptyset$) to the Pachner $(\frac{n}{2}-1)$-$\frac{n}{2}$ move for the $(n-1)$-gon $a_1a_2a_3\cdots a_{n-1}$ (regarded as a simplicial complex) that replaces the $(n-3)$-simplices $a_1a_2a_3\cdots a_{n-4}a_{n-2}a_{n-1}$, 
$a_1a_2a_3\cdots a_{n-6}a_{n-4}a_{n-3}a_{n-2}a_{n-1}$, 
$\dots$, 
$a_2a_3a_4\cdots a_{n-2}a_{n-1}$ 
with the $(n-3)$-simplices $a_1a_2a_3\cdots a_{n-2}$, 
$a_1a_2a_3\cdots a_{n-3}a_{n-1}$, 
$\dots$, 
$a_1a_3a_4\cdots a_{n-1}$ as follows: the $(i,j)$-entry of the matrix $P_{a_1a_3a_5\cdots a_{n-3},a_2a_4a_6\cdots a_{n-2}a_{n-1}}$ is
\bea\label{eq:P-even}
(P_{a_1a_3a_5\cdots a_{n-3},a_2a_4a_6\cdots a_{n-2}a_{n-1}})_i^j=\left\{\begin{array}{ll}(-1)^{j+\frac{n}{2}-1}\frac{V(a_{n+2-2i},a_1,a_3,\dots,\hat{a}_{n-1-2j},\dots,a_{n-3})}{V(a_1,a_3,a_5,\dots,a_{n-3})}&\text{if }i>1\\
(-1)^{j+\frac{n}{2}-1}\frac{V(a_{n-1},a_1,a_3,\dots,\hat{a}_{n-1-2j},\dots,a_{n-3})}{V(a_1,a_3,a_5,\dots,a_{n-3})}&\text{if }i=1
\end{array}\right.
\eea
where $\hat{a}_k$ means that $a_k$ is omitted.

Note that the sum of the entries of each row of the matrix $P$ is 1.
This construction is obtained by induction on $n$ and inspired by \cite{KS13}, see later sections for examples.

We can extend the matrices we constructed for Pachner moves above by adding a row and a column so that the entry is 1 at the intersection and 0 elsewhere corresponding to each fixed $(n-3)$-simplex in the Pachner moves.

\begin{theorem}\label{main}
The extended matrices satisfy the $n$-gon equation \eqref{eq:n-gon-odd} or \eqref{eq:n-gon-even}.
\end{theorem}
\subsection{Comparison of our solution to the odd-gon equation with the one in \cite{DK21}}
Note that our solution to the odd-gon equation is exactly a partial case of the one in \cite{DK21}. To show this, we introduce Grassmannian and Pl{\" u}cker coordinates.
Let $F$ be a field, $F^{n}$ the row vector space of dimension $n$ over $F$, $\mathsf e_1,\ldots  ,\mathsf e_{n}$ the standard basis in $F^{n}$. 

We assume that $n$ is odd.
Let
\begin{equation}
\mathcal{M}= \begin{pmatrix}
 \alpha _{1,1} & \alpha _{1,2} & \cdots  & \alpha _{1,n} \\
 \alpha _{2,1} & \alpha _{2,2} & \cdots  & \alpha _{2,n} \\
 \vdots  & \vdots  & \ddots & \vdots  \\
 \alpha _{{\frac{n-1}{2}},1} & \alpha _{{\frac{n-1}{2}},2} & \cdots  & \alpha _{{\frac{n-1}{2}},n} \\
\end{pmatrix} 
\end{equation}
and
\begin{equation}
\mathcal{M}^{\perp}= \begin{pmatrix}
 \beta _{1,1} & \beta _{1,2} & \cdots  & \beta _{1,n} \\
 \beta _{2,1} & \beta _{2,2} & \cdots  & \beta _{2,n} \\
 \vdots  & \vdots  & \ddots & \vdots  \\
 \beta _{{\frac{n+1}{2}},1} & \beta _{{\frac{n+1}{2}},2} & \cdots  & \beta _{{\frac{n+1}{2}},n} \\
\end{pmatrix} 
\end{equation}
be matrices defining elements $L$ of the Grassmannian $\Gr({\frac{n-1}{2}},n)$ and $L^{\perp}$ of the Grassmannian $\Gr({\frac{n+1}{2}},n)$ where $L^{\perp}$ is the orthogonal complement of $L$, that is, matrices whose rows are of the full rank ${\frac{n-1}{2}}$ and ${\frac{n+1}{2}}$ respectively, which can also be written as
\[
v_i=\beta _{i,1}\mathsf e_1+\dots+\beta _{i,n}\mathsf e_{n},
\]
span an ${\frac{n+1}{2}}$-dimensional plane $L^{\perp}\subset F^{n}$.

The \emph{multivectors} are elements of the \emph{exterior algebra} $\bigwedge F^{n}$ over $F^{n}$. We introduce an ${\frac{n+1}{2}}$-vector---the exterior product of all $v_i$:
\[
w =v_1\wedge v_2\wedge \cdots \wedge v_{{\frac{n+1}{2}}}=\sum _{k_1<\cdots <k_{{\frac{n+1}{2}}}} p_{k_1,\ldots ,k_{{\frac{n+1}{2}}}}(L^{\perp})\mathsf e_{k_1}
\wedge \dots \wedge \mathsf e_{k_{{\frac{n+1}{2}}}},
\]
where $p_{k_1,\ldots ,k_{{\frac{n+1}{2}}}}(L^{\perp})$ are determinants made of $k_1$-th, \ldots, $k_{{\frac{n+1}{2}}}$-th columns of $\mathcal{M}^{\perp}$, called also \emph{Pl{\" u}cker coordinates} of the Grassmannian $\Gr({\frac{n+1}{2}},n)$.

\begin{assumption}\label{a:p-nonzero}
Below, we assume that all Pl{\" u}cker coordinates are nonzero.
\end{assumption}
In \Sec{sec:proof}, we will consider a partial case: Grassmannian defined by the Vandermonde matrix, which satisfies Assumption \ref{a:p-nonzero}.

The Pl{\" u}cker coordinates define an embedding, called \emph{Pl{\" u}cker embedding}, of the Grassmannian $\Gr({\frac{n+1}{2}},n)$ into the projectivization $\mathbb{P}(\bigwedge^{{\frac{n+1}{2}}}F^n)$ of the ${\frac{n+1}{2}}$-th exterior power of $F^n$. Note that the Pl{\" u}cker coordinates are homogeneous coordinates, we may multiply them by a common factor.

\begin{lemma}\label{dual}
Let $\{l_1,\dots,l_{{\frac{n-1}{2}}}\}=\{1,2,\dots,n\}\setminus\{k_1,\ldots ,k_{{\frac{n+1}{2}}}\}$ and $L^{\perp}$ be the orthogonal complement of $L$ in $F^n$, then the Pl{\" u}cker coordinates of $L$ and $L^{\perp}$ are related as follows: 
\bea
p_{k_1,\ldots ,k_{{\frac{n+1}{2}}}}(L^{\perp})=(-1)^{\#(\{l_1,\dots,l_{{\frac{n-1}{2}}}\}\cap\{1,2,\dots,{\frac{n+1}{2}}\})}p_{l_1,\dots,l_{{\frac{n-1}{2}}}}(L)
\eea
up to a common factor.
\end{lemma}
\begin{proof}
We write $\mathcal{M}=(A,B)$ where $A$ is a ${\frac{n+1}{2}}\times{\frac{n+1}{2}}$ matrix and $B$ is a ${\frac{n+1}{2}}\times{\frac{n-1}{2}}$ matrix. Since $A$ is invertible, we can multiply $\mathcal{M}$ by $A^{-1}$ from the left, this changes the Pl{\" u}cker coordinates by a common factor. Without loss of generality, we may assume that $A=I$ is the identity ${\frac{n+1}{2}}\times{\frac{n+1}{2}}$ matrix and $\mathcal{M}=(I,B)$. The matrix defining $L^{\perp}$ can be chosen as the transpose of $\begin{pmatrix}-B\\I\end{pmatrix}$ where $I$ is the identity ${\frac{n-1}{2}}\times{\frac{n-1}{2}}$ matrix. Then a direct computation proves this lemma.   
\end{proof}

We rewrite the $(i,j)$-entry of the matrix in \cite[eq. (44)]{DK21} in our notation as
\bea
(-1)^i\frac{p_{a_{2j},a_1,a_3,\dots,\hat{a}_{2i-1},\dots,a_{n-2},q}(L^{\perp})}{p_{a_1,a_3,\dots,a_{n-2},q}(L^{\perp})}
\eea
where $\{1,2,\dots,n\}\setminus\{q\}=\{a_1,a_2,\dots,a_{n-1}\}$.
We compare this with \eqref{eq:P-odd} and use Lemma \ref{dual}, we find that our solution to the odd-gon equation is exactly a partial case of the one in \cite{DK21}.

\section{Proof of the main results}\label{sec:proof}
In this section, we prove Theorem \ref{main} by the following steps:
\begin{enumerate}
\item[Step 1]
We associate a vector $f^{(n)}_{1\cdots\hat{i}\cdots\hat{j}\cdots n}$ to each $(n-3)$-simplex where $i$ and $j$ are not vertices of the $(n-3)$-simplex.
\item[Step 2]
The $n$-gon equation involves two sequences of Pachner moves. Each sequence is a path from the initial triangulation to the final triangulation of the $n$-gon. The number of $(n-3)$-simplices in each step of the two paths does not change for odd $n$ and increases by 1 for even $n$.

We can extend the matrices we constructed for Pachner moves in \Sec{sec:matrices} by adding a row and a column so that the entry is 1 at the intersection and 0 elsewhere corresponding to each fixed $(n-3)$-simplex in the Pachner moves.

By Proposition \ref{matrix}, the extended matrices map the vectors associated to $(n-3)$-simplices in the triangulation in each step to the vectors associated to $(n-3)$-simplices in the triangulation in the next step. So the products of matrices on both sides of the $n$-gon equation act on the vectors associated to $(n-3)$-simplices in the initial triangulation in the same way. 
\item[Step 3]
By Proposition \ref{f-polygon}, the vectors $f^{(n)}_{1\cdots\hat{i}\cdots\hat{j}\cdots n}$ associated to $(n-3)$-simplices in the initial triangulation are linearly independent, so we conclude that the products of matrices on both sides of the $n$-gon equation are equal.
\end{enumerate}

 \subsection{Construction of vectors associated to $(n-3)$-simplices}
In this subsection, we construct vectors associated to $(n-3)$-simplices. This construction is inspired by \cite{KS13}.

Let 
\bea\label{eq:M}
\mathcal{M}=\begin{pmatrix}1&1&\cdots&1\\\zeta_1&\zeta_2&\cdots&\zeta_n\\\zeta_1^2&\zeta_2^2&\cdots&\zeta_n^2\\\vdots&\vdots&\ddots&\vdots\\\zeta_1^{\floor{\frac{n}{2}}-1}&\zeta_2^{\floor{\frac{n}{2}}-1}&\cdots&\zeta_n^{\floor{\frac{n}{2}}-1}\end{pmatrix},
\eea
then the Pl{\" u}cker coordinates are the Vandermonde determinants $V(k_1,\ldots ,k_{\floor{\frac{n}{2}}})$ which are nonzero.

Given $n$ pairwise distinct variables $\zeta_i$ ($i=1,2,\dots,n$) and a $(n-3)$-simplex $a_1a_2\cdots a_{n-2}$ ($1\le a_1<a_2<\cdots<a_{n-2}\le n$), we can construct a vector $(f_{a_1}^{(n)},f_{a_2}^{(n)},\dots,f_{a_{n-2}}^{(n)})$ such that 
\bea\label{eq:vector}
\sum_{i=1}^{n-2}f_{a_i}^{(n)}\zeta_{a_i}^m=0\text{ for }m=0,1,\dots,\floor{\frac{n}{2}}-1
\eea
where $\floor{x}$ means the largest integer less than or equal to $x$.
For example, 
we can choose 
\bea
f_{a_i}^{(n)}=f^{(n)}(a_i,a_1, a_2,\dots, \hat{a}_i,\dots,a_{n-2})
\eea
where $\hat{a}_i$ means that $a_i$ is omitted and $f^{(n)}$ is a to-be-defined function. To define $f^{(n)}$, we first define a function $g^{(n)}$ as follows:
\bea
g_{a_i}^{(n)}=g^{(n)}(a_i,a_1, a_2,\dots, \hat{a}_i,\dots,a_{n-2})
\eea
and
\bea
g^{(n)}(a_1,a_2,\dots,a_{n-2})=\frac{1}{\prod\limits_{i=2}^{n-2}(\zeta_{a_1}-\zeta_{a_i})}.
\eea
Note that $g^{(n)}$ is symmetric with respect to the variables $a_i$ for $2\le i\le n-2$.
Then we define $f^{(n)}$ as follows:
\bea
f^{(n)}(a_1,a_2,\dots,a_{n-2})=\sum_{2\le i_1<i_2<\cdots<i_{\floor{\frac{n}{2}}}\le n-2}g^{(\floor{\frac{n}{2}}+3)}(a_1,a_{i_1},a_{i_2},\dots,a_{i_{\floor{\frac{n}{2}}}})
\eea
where $\floor{x}$ means the largest integer less than or equal to $x$.
For example,
\bea
f^{(5)}(a_1,a_2,a_3)&=&g^{(5)}(a_1,a_2,a_3),\nn\\
f^{(6)}(a_1,a_2,a_3,a_4)&=&g^{(6)}(a_1,a_2,a_3,a_4),\nn\\
f^{(7)}(a_1,a_2,a_3,a_4,a_5)&=&\sum_{i=2}^5g^{(6)}(a_1,a_2,\dots,\hat{a}_i,\dots,a_5),\nn\\
f^{(8)}(a_1,a_2,a_3,a_4,a_5,a_6)&=&\sum_{i=2}^6g^{(7)}(a_1,a_2,\dots,\hat{a}_i,\dots,a_6)
\eea
where $\hat{a}_i$ means that $a_i$ is omitted.

\begin{lemma}\label{g}
The function $g^{(n)}$ defined above satisfies
\bea
\sum_{i=1}^{n-2}g_{a_i}^{(n)}\zeta_{a_i}^m=0\text{ for }m=0,1,\dots,n-4.
\eea
\end{lemma}
\begin{proof}
Note that 
\bea
g^{(n)}(a_i,a_1, a_2,\dots, \hat{a}_i,\dots,a_{n-2})=(-1)^{i+1}\frac{\prod\limits_{\begin{array}{cc}1\le j<k\le n-2\\j,k\ne i\end{array}}(\zeta_{a_j}-\zeta_{a_k})}{\prod\limits_{1\le j<k\le n-2}(\zeta_{a_j}-\zeta_{a_k})}
\eea
where both the denominator and the numerator are Vandermonde determinants. 
Also, note that
\bea
\begin{vmatrix}\zeta_{a_1}^m&\zeta_{a_2}^m&\cdots&\zeta_{a_{n-2}}^m\\
1&1&\cdots&1\\
\zeta_{a_1}&\zeta_{a_2}&\cdots&\zeta_{a_{n-2}}\\
\zeta_{a_1}^2&\zeta_{a_2}^2&\cdots&\zeta_{a_{n-2}}^2\\
\vdots&\vdots&\ddots&\vdots\\
\zeta_{a_1}^{n-4}&\zeta_{a_2}^{n-4}&\cdots&\zeta_{a_{n-2}}^{n-4}\\
\end{vmatrix}=0
\eea
for $m=0,1,\dots,n-4$.
Then expanding this determinant along the first row proves this lemma. 
\end{proof}

\begin{proposition}\label{f}
The function $f^{(n)}$ defined above satisfies \eqref{eq:vector}.
\end{proposition}
\begin{proof}
By definition,
\bea
&&f^{(n)}(a_i,a_1, a_2,\dots, \hat{a}_i,\dots,a_{n-2})\nn\\
&=&\sum_{\begin{array}{cc}1\le i_1<i_2<\cdots<i_{\floor{\frac{n}{2}}}\le n-2\\i_l\ne i,\forall l\end{array}}g^{(\floor{\frac{n}{2}}+3)}(a_i,a_{i_1},a_{i_2},\dots,a_{i_{\floor{\frac{n}{2}}}})
\eea
where $\floor{x}$ means the largest integer less than or equal to $x$.
Therefore,
\bea
&&\sum_{i=1}^{n-2}f^{(n)}(a_i,a_1, a_2,\dots, \hat{a}_i,\dots,a_{n-2})\zeta_{a_i}^m\nn\\
&=&\sum_{i=1}^{n-2}\sum_{\begin{array}{cc}1\le i_1<i_2<\cdots<i_{\floor{\frac{n}{2}}}\le n-2\\i_l\ne i,\forall l\end{array}}g^{(\floor{\frac{n}{2}}+3)}(a_i,a_{i_1},a_{i_2},\dots,a_{i_{\floor{\frac{n}{2}}}})\zeta_{a_i}^m\nn\\
&=&\sum_{1\le j_1<j_2<\cdots<j_{\floor{\frac{n}{2}}+1}\le n-2}\sum_{k=1}^{\floor{\frac{n}{2}}+1}g^{(\floor{\frac{n}{2}}+3)}(a_{j_k},a_{j_1},\dots,\hat{a}_{j_k},\dots,a_{j_{\floor{\frac{n}{2}}+1}})\zeta_{a_{j_k}}^m\nn\\
&=&0
\eea
where $\hat{a}_i$ means that $a_i$ is omitted, $\{j_1,j_2,\dots,j_{\floor{\frac{n}{2}}+1}\}=\{i,i_1,i_2,\dots, i_{\floor{\frac{n}{2}}}\}$, and we have used Lemma \ref{g} in the last equality.
Hence we have proved that the function $f^{(n)}$ satisfies \eqref{eq:vector}.
\end{proof}

We add zeros in the vector $(f_{a_1}^{(n)},f_{a_2}^{(n)},\dots,f_{a_{n-2}}^{(n)})$ at the positions which are not the vertices of the $(n-3)$-simplex to extend the vector $(f_{a_1}^{(n)},f_{a_2}^{(n)},\dots,f_{a_{n-2}}^{(n)})$ to a vector of length $n$. We denote the extended vector by $f^{(n)}_{a_1a_2\cdots a_{n-2}}$.

Proposition \ref{f} implies that the vectors $f^{(n)}_{a_1a_2\cdots a_{n-2}}$ are orthogonal to the row vectors of $\mathcal{M}$ \eqref{eq:M}, namely the vectors $f^{(n)}_{a_1a_2\cdots a_{n-2}}$ are in $L^{\perp}$ where $L$ is the $\floor{\frac{n}{2}}$-plane in $F^n$ defined by $\mathcal{M}$.

 Each Pachner move in the $n$-gon equation only involves $n-1$ vertices, we denote the remaining vertex by $q$. Let $\{1,2,\dots,n\}\setminus\{q\}=\{a_1,a_2,\dots,a_{n-1}\}$.

\begin{proposition}\label{linear-independent}
For any $\floor{\frac{n-1}{2}}$ different values $j_1$, $\dots$, $j_{\floor{\frac{n-1}{2}}}$ in $\{1,2,\dots,n-1\}$, the vectors $f^{(n)}_{a_1\cdots \hat{a}_{j_m}\cdots a_{n-1}}$ $(m=1,\dots,\floor{\frac{n-1}{2}})$ are linearly independent.
\end{proposition}
\begin{proof}
By definition, the $k$-th component of the vector $f^{(n)}_{a_1a_2\cdots a_{n-2}}$ is
\bea
(f^{(n)}_{a_1a_2\cdots a_{n-2}})^k=\left\{\begin{array}{ll}f^{(n)}(a_i,a_1,\dots,\hat{a}_i,\dots,a_{n-2})&\text{if }k=a_i\\0&\text{if }k\not\in\{a_1,\dots,a_{n-2}\}\end{array}\right.
\eea
and
\bea
f^{(n)}(a_i,a_1,\dots,\hat{a}_i,\dots,a_{n-2})&=&\sum_{\begin{array}{cc}1\le i_1<i_2<\cdots<i_{\floor{\frac{n}{2}}}\le n-2\\i_l\ne i,\forall l\end{array}}\frac{V(a_{i_1},a_{i_2},\dots,a_{i_{\floor{\frac{n}{2}}}})}{V(a_i,a_{i_1},a_{i_2},\dots,a_{i_{\floor{\frac{n}{2}}}})}\nn\\
&=&\sum_{\begin{array}{cc}1\le i_1<i_2<\cdots<i_{\floor{\frac{n}{2}}}\le n-2\\i_l\ne i,\forall l\end{array}}\frac{1}{\prod\limits_{l=1}^{\floor{\frac{n}{2}}}(\zeta_{a_i}-\zeta_{a_{i_l}})}
\eea
where $V(a_{i_1},a_{i_2},\dots,a_{i_l})$ is the Vandermonde determinant \eqref{eq:Vandermonde} and $\floor{x}$ means the largest integer less than or equal to $x$.

We need to prove that if 
\bea
\sum_{m=1}^{\floor{\frac{n-1}{2}}}\lambda_mf^{(n)}_{a_1\cdots \hat{a}_{j_m}\cdots a_{n-1}}=0,
\eea
then $\lambda_m=0$ for $m=1,\dots,\floor{\frac{n-1}{2}}$.
The $j_{m_0}$-th component is
\bea
\sum_{\begin{array}{cc}m=1\\m\ne m_0\end{array}}^{\floor{\frac{n-1}{2}}}\lambda_m\sum_{\begin{array}{cc}1\le i_1<i_2<\cdots<i_{\floor{\frac{n}{2}}}\le n-1\\i_l\ne j_m,j_{m_0},\forall l\end{array}}\frac{1}{\prod\limits_{l=1}^{\floor{\frac{n}{2}}}(\zeta_{a_{j_{m_0}}}-\zeta_{a_{i_l}})}=0.
\eea
The coefficient of the term corresponding to $\{i_1,\dots,i_{\floor{\frac{n}{2}}}\}=\{1,2,\dots,n-1\}\setminus\{j_1,\dots,j_{\floor{\frac{n-1}{2}}}\}$ is 
\bea
\sum_{\begin{array}{cc}m=1\\m\ne m_0\end{array}}^{\floor{\frac{n-1}{2}}}\lambda_m=0.
\eea
Hence $\lambda_m=0$ for $m=1,\dots,\floor{\frac{n-1}{2}}$. We have used the assumption that the characteristic of $F$ is zero.
\end{proof}
\begin{proposition}\label{exactly}
The vectors $f^{(n)}_{a_1a_2\cdots \hat{a}_i\cdots a_{n-1}}$ ($i=1,2,\dots,n-1$) span a $\floor{\frac{n-1}{2}}$-dimensional linear space.
\end{proposition}
\begin{proof}
By Proposition \ref{f}, the vectors $f^{(n)}_{a_1a_2\cdots \hat{a}_i\cdots a_{n-1}}$ are in $L^{\perp}$ where $L$ is the $\floor{\frac{n}{2}}$-plane in $F^n$ defined by $\mathcal{M}$ \eqref{eq:M}. Also, by definition, the $q$-th component of the vectors $f^{(n)}_{a_1a_2\cdots \hat{a}_i\cdots a_{n-1}}$ is zero. So the vectors $f^{(n)}_{a_1a_2\cdots \hat{a}_i\cdots a_{n-1}}$ are in $(L\oplus\mathsf{e}_q)^{\perp}$ where $\mathsf{e}_q$ is the $q$-th standard basis of $F^n$. Therefore, the dimension of the linear space spanned by the vectors $f^{(n)}_{a_1a_2\cdots \hat{a}_i\cdots a_{n-1}}$ is less than or equal to $n-1-\floor{\frac{n}{2}}=\floor{\frac{n-1}{2}}$. By Proposition \ref{linear-independent}, the dimension of the linear space spanned by the vectors $f^{(n)}_{a_1a_2\cdots \hat{a}_i\cdots a_{n-1}}$ is exactly $\floor{\frac{n-1}{2}}$.
\end{proof}

By Proposition \ref{exactly}, the vectors $f^{(n)}_{a_1a_2\cdots \hat{a}_i\cdots a_{n-1}}$ associated to the $(n-3)$-simplices in the triangulation after a Pachner move are linear combinations of the vectors $f^{(n)}_{a_1a_2\cdots \hat{a}_i\cdots a_{n-1}}$ associated to the $(n-3)$-simplices in the triangulation before that Pachner move. 
\begin{proposition}\label{matrix}
For odd $n$, 
\bea
P_{a_2a_4a_6\cdots a_{n-1},a_1a_3a_5\cdots a_{n-2}}\cdot\begin{pmatrix} 
f^{(n)}_{a_1a_2a_3\cdots a_{n-2}}\\
f^{(n)}_{a_1a_2a_3\cdots a_{n-4}a_{n-2}a_{n-1}}\\
\vdots\\
f^{(n)}_{a_1a_3a_4\cdots a_{n-2}a_{n-1}}
\end{pmatrix}
=\begin{pmatrix} 
f^{(n)}_{a_1a_2a_3\cdots a_{n-3}a_{n-1}}\\
f^{(n)}_{a_1a_2a_3\cdots a_{n-5}a_{n-3}a_{n-2}a_{n-1}}\\
\vdots\\
f^{(n)}_{a_2a_3a_4\cdots a_{n-1}}
\end{pmatrix}.
\eea
For even $n$,
\bea
P_{a_1a_3a_5\cdots a_{n-3},a_2a_4a_6\cdots a_{n-2}a_{n-1}}\cdot\begin{pmatrix} 
f^{(n)}_{a_1a_2a_3\cdots a_{n-4}a_{n-2}a_{n-1}}\\
f^{(n)}_{a_1a_2a_3\cdots a_{n-6}a_{n-4}a_{n-3}a_{n-2}a_{n-1}}\\
\vdots\\
f^{(n)}_{a_2a_3a_4\cdots a_{n-2}a_{n-1}}
\end{pmatrix}
=\begin{pmatrix} 
f^{(n)}_{a_1a_2a_3\cdots a_{n-2}}\\
f^{(n)}_{a_1a_2a_3\cdots a_{n-3}a_{n-1}}\\
\vdots\\
f^{(n)}_{a_1a_3a_4\cdots a_{n-1}}
\end{pmatrix}.
\eea
\end{proposition}
\begin{proof}
For odd $n$, 
note that
\bea
(P_{a_2a_4a_6\cdots a_{n-1},a_1a_3a_5\cdots a_{n-2}})_i^j&=&(-1)^{j+\frac{n-1}{2}}\frac{V(a_{n-2i},a_2,a_4,\dots,\hat{a}_{n+1-2j},\dots,a_{n-1})}{V(a_2,a_4,a_6,\dots,a_{n-1})}\nn\\&=&\frac{\prod\limits_{\begin{array}{cc}j'=1\\j'\ne j\end{array}}^{\frac{n-1}{2}}(\zeta_{a_{n-2i}}-\zeta_{a_{n+1-2j'}})}{\prod\limits_{\begin{array}{cc}j'=1\\j'\ne j\end{array}}^{\frac{n-1}{2}}(\zeta_{a_{n+1-2j}}-\zeta_{a_{n+1-2j'}})},
\eea
we need to prove that
\bea
&&\sum_{j=1}^{\frac{n-1}{2}}\frac{\prod\limits_{\begin{array}{cc}j'=1\\j'\ne j\end{array}}^{\frac{n-1}{2}}(\zeta_{a_{n-2i}}-\zeta_{a_{n+1-2j'}})}{\prod\limits_{\begin{array}{cc}j'=1\\j'\ne j\end{array}}^{\frac{n-1}{2}}(\zeta_{a_{n+1-2j}}-\zeta_{a_{n+1-2j'}})}
f^{(n)}_{a_1a_2\cdots\hat{a}_{n+1-2j}\cdots a_{n-1}}\nn\\
&=&f^{(n)}_{a_1\cdots\hat{a}_{n-2i}\cdots a_{n-1}}
\eea
where $\hat{a}_i$ means that $a_i$ is omitted.
In particular, the $a_{n-2i}$-th component is
\bea\label{eq:odd-1}
\sum_{j=1}^{\frac{n-1}{2}}\frac{\prod\limits_{\begin{array}{cc}j'=1\\j'\ne j\end{array}}^{\frac{n-1}{2}}(\zeta_{a_{n-2i}}-\zeta_{a_{n+1-2j'}})}{\prod\limits_{\begin{array}{cc}j'=1\\j'\ne j\end{array}}^{\frac{n-1}{2}}(\zeta_{a_{n+1-2j}}-\zeta_{a_{n+1-2j'}})}
\sum_{\begin{array}{cc}1\le i_1<i_2<\cdots<i_{{\frac{n-1}{2}}}\le n-1\\i_l\ne n+1-2j,n-2i,\forall l\end{array}}\frac{1}{\prod\limits_{l=1}^{{\frac{n-1}{2}}}(\zeta_{a_{n-2i}}-\zeta_{a_{i_l}})}=0.
\eea
The $a_{n-2i'}$-th component ($i'\ne i$) is
 \bea\label{eq:odd-2}
&&\sum_{j=1}^{\frac{n-1}{2}}\frac{\prod\limits_{\begin{array}{cc}j'=1\\j'\ne j\end{array}}^{\frac{n-1}{2}}(\zeta_{a_{n-2i}}-\zeta_{a_{n+1-2j'}})}{\prod\limits_{\begin{array}{cc}j'=1\\j'\ne j\end{array}}^{\frac{n-1}{2}}(\zeta_{a_{n+1-2j}}-\zeta_{a_{n+1-2j'}})}
\sum_{\begin{array}{cc}1\le i_1<i_2<\cdots<i_{{\frac{n-1}{2}}}\le n-1\\i_l\ne n+1-2j,n-2i',\forall l\end{array}}\frac{1}{\prod\limits_{l=1}^{{\frac{n-1}{2}}}(\zeta_{a_{n-2i'}}-\zeta_{a_{i_l}})}\nn\\
&=&\sum_{\begin{array}{cc}1\le i_1<i_2<\cdots<i_{{\frac{n-1}{2}}}\le n-1\\i_l\ne n-2i,n-2i',\forall l\end{array}}\frac{1}{\prod\limits_{l=1}^{{\frac{n-1}{2}}}(\zeta_{a_{n-2i'}}-\zeta_{a_{i_l}})}.
\eea
The $a_{n+1-2j_0}$-th component is
 \bea\label{eq:odd-3}
&&\sum_{\begin{array}{cc}j=1\\j\ne j_0\end{array}}^{\frac{n-1}{2}}\frac{\prod\limits_{\begin{array}{cc}j'=1\\j'\ne j\end{array}}^{\frac{n-1}{2}}(\zeta_{a_{n-2i}}-\zeta_{a_{n+1-2j'}})}{\prod\limits_{\begin{array}{cc}j'=1\\j'\ne j\end{array}}^{\frac{n-1}{2}}(\zeta_{a_{n+1-2j}}-\zeta_{a_{n+1-2j'}})}
\sum_{\begin{array}{cc}1\le i_1<i_2<\cdots<i_{{\frac{n-1}{2}}}\le n-1\\i_l\ne n+1-2j,n+1-2j_0,\forall l\end{array}}\frac{1}{\prod\limits_{l=1}^{{\frac{n-1}{2}}}(\zeta_{a_{n+1-2j_0}}-\zeta_{a_{i_l}})}\nn\\
&=&\sum_{\begin{array}{cc}1\le i_1<i_2<\cdots<i_{{\frac{n-1}{2}}}\le n-1\\i_l\ne n-2i,n+1-2j_0,\forall l\end{array}}\frac{1}{\prod\limits_{l=1}^{{\frac{n-1}{2}}}(\zeta_{a_{n+1-2j_0}}-\zeta_{a_{i_l}})}.
\eea
We can use the fact 
\bea
\sum_{j=1}^{\frac{n-1}{2}}\frac{1}{\prod\limits_{\begin{array}{cc}j'=1\\j'\ne j\end{array}}^{\frac{n-1}{2}}(\zeta_{a_{n+1-2j}}-\zeta_{a_{n+1-2j'}})}=0
\eea
to prove \eqref{eq:odd-1}.
We can use the trick
\bea
\frac{\zeta_{a_{n-2i}}-\zeta_{a_{n+1-2j'}}}{(\zeta_{a_{n-2i'}}-\zeta_{a_{n-2i}})(\zeta_{a_{n-2i'}}-\zeta_{a_{n+1-2j'}})}=\frac{1}{\zeta_{a_{n-2i'}}-\zeta_{a_{n-2i}}}-\frac{1}{\zeta_{a_{n-2i'}}-\zeta_{a_{n+1-2j'}}}
\eea
to prove \eqref{eq:odd-2}.
The proof for \eqref{eq:odd-3} is straightforward. 

For even $n$, 
note that
\bea
(P_{a_1a_3a_5\cdots a_{n-3},a_2a_4a_6\cdots a_{n-2}a_{n-1}})_i^j&=&\left\{\begin{array}{ll}(-1)^{j+\frac{n}{2}-1}\frac{V(a_{n+2-2i},a_1,a_3,\dots,\hat{a}_{n-1-2j},\dots,a_{n-3})}{V(a_1,a_3,a_5,\dots,a_{n-3})}&\text{if }i>1\\
(-1)^{j+\frac{n}{2}-1}\frac{V(a_{n-1},a_1,a_3,\dots,\hat{a}_{n-1-2j},\dots,a_{n-3})}{V(a_1,a_3,a_5,\dots,a_{n-3})}&\text{if }i=1
\end{array}\right.\nn\\
&=&\left\{\begin{array}{ll}\frac{\prod\limits_{\begin{array}{cc}j'=1\\j'\ne j\end{array}}^{\frac{n}{2}-1}(\zeta_{a_{n+2-2i}}-\zeta_{a_{n-1-2j'}})}{\prod\limits_{\begin{array}{cc}j'=1\\j'\ne j\end{array}}^{\frac{n}{2}-1}(\zeta_{a_{n-1-2j}}-\zeta_{a_{n-1-2j'}})}&\text{if }i>1\\
\frac{\prod\limits_{\begin{array}{cc}j'=1\\j'\ne j\end{array}}^{\frac{n}{2}-1}(\zeta_{a_{n-1}}-\zeta_{a_{n-1-2j'}})}{\prod\limits_{\begin{array}{cc}j'=1\\j'\ne j\end{array}}^{\frac{n}{2}-1}(\zeta_{a_{n-1-2j}}-\zeta_{a_{n-1-2j'}})}&\text{if }i=1
\end{array}\right.,
\eea
we need to prove that for $i>1$,
\bea
&&\sum_{j=1}^{\frac{n}{2}-1}\frac{\prod\limits_{\begin{array}{cc}j'=1\\j'\ne j\end{array}}^{\frac{n}{2}-1}(\zeta_{a_{n+2-2i}}-\zeta_{a_{n-1-2j'}})}{\prod\limits_{\begin{array}{cc}j'=1\\j'\ne j\end{array}}^{\frac{n}{2}-1}(\zeta_{a_{n-1-2j}}-\zeta_{a_{n-1-2j'}})}
f^{(n)}_{a_1\cdots\hat{a}_{n-1-2j}\cdots a_{n-1}}\nn\\
&=&f^{(n)}_{a_1\cdots\hat{a}_{n+2-2i}\cdots a_{n-1}},
\eea
and 
\bea
&&\sum_{j=1}^{\frac{n}{2}-1}\frac{\prod\limits_{\begin{array}{cc}j'=1\\j'\ne j\end{array}}^{\frac{n}{2}-1}(\zeta_{a_{n-1}}-\zeta_{a_{n-1-2j'}})}{\prod\limits_{\begin{array}{cc}j'=1\\j'\ne j\end{array}}^{\frac{n}{2}-1}(\zeta_{a_{n-1-2j}}-\zeta_{a_{n-1-2j'}})}
f^{(n)}_{a_1\cdots\hat{a}_{n-1-2j}\cdots a_{n-1}}\nn\\
&=&f^{(n)}_{a_1\cdots a_{n-2}}
\eea
where $\hat{a}_i$ means that $a_i$ is omitted. The proof is similar to that for odd $n$.
\end{proof}

The $n$-gon equation involves two sequences of Pachner moves. Each sequence is a path from the initial triangulation to the final triangulation of the $n$-gon. The number of $(n-3)$-simplices in each step of the two paths does not change for odd $n$ and increases by 1 for even $n$.

 We can extend the matrices we constructed for Pachner moves in \Sec{sec:matrices} by adding a row and a column so that the entry is 1 at the intersection and 0 elsewhere corresponding to each fixed $(n-3)$-simplex in the Pachner moves.

By Proposition \ref{matrix}, the extended matrices map the vectors associated to $(n-3)$-simplices in the triangulation in each step to the vectors associated to $(n-3)$-simplices in the triangulation in the next step. So the products of matrices on both sides of the $n$-gon equation act on the vectors associated to $(n-3)$-simplices in the initial triangulation in the same way.

\begin{proposition}\label{f-polygon}
The vectors $f^{(n)}_{1\cdots\hat{i}\cdots\hat{j}\cdots n}$ associated to $(n-3)$-simplices in the initial triangulation of the $n$-gon \eqref{eq:pairs} are linearly independent, that is, span an $\dfrac{\floor{\frac{n-1}{2}}(\floor{\frac{n-1}{2}}+1)}{2}$-dimensional linear space.
\end{proposition}
\begin{proof}
By definition, 
\bea
f^{(n)}_{1\cdots\hat{i}\cdots\hat{j}\cdots n}=\sum_{\begin{array}{cc}r=1\\r\ne i,j\end{array}}^n\mathsf{e}_r\sum_{\begin{array}{cc}1\le i_1<\cdots<i_{\floor{\frac{n}{2}}}\le n\\i_m\ne i,j,r,\forall m\end{array}}\frac{1}{\prod\limits_{m=1}^{\floor{\frac{n}{2}}}(\zeta_{r}-\zeta_{{i_m}})}
\eea
where $\mathsf{e}_r$ ($r=1,2,\dots,n$) are the standard basis of $F^n$. 
We need to prove that if 
\bea
\sum_{k=1}^{\floor{\frac{n-1}{2}}}\sum_{l=1}^k\lambda_{n+1-2k,n+2-2l} f^{(n)}_{1\cdots\widehat{n+1-2k}\cdots\widehat{n+2-2l}\cdots n}=0,
\eea
then $\lambda_{n+1-2k,n+2-2l}=0$ for $1\le l\le k\le \floor{\frac{n-1}{2}}$.
For each fixed $l_0$, the coefficient of $\mathsf{e}_{n+2-2l_0}$ is
\bea
\sum_{k=1}^{\floor{\frac{n-1}{2}}}\sum_{\begin{array}{cc}l=1\\l\ne l_0\end{array}}^k\lambda_{n+1-2k,n+2-2l}\sum_{\begin{array}{cc}1\le i_1<\cdots<i_{\floor{\frac{n}{2}}}\le n\\i_m\ne n+1-2k,n+2-2l,n+2-2l_0,\forall m\end{array}}\frac{1}{\prod\limits_{m=1}^{\floor{\frac{n}{2}}}(\zeta_{n+2-2l_0}-\zeta_{{i_m}})}.
\eea
For each fixed $k_0$, the coefficient of the term corresponding to $\{i_1,\dots,i_{\floor{\frac{n}{2}}}\}=\{1,2,\dots,n\}\setminus(\{n+2-2l\mid l=1,2,\dots,\floor{\frac{n-1}{2}}\}\cup\{n+1-2k_0\})$ in the coefficient of $\mathsf{e}_{n+2-2l_0}$ is 
\bea\label{eq:lambda-1}
\sum_{\begin{array}{cc}l=1\\l\ne l_0\end{array}}^{k_0}\lambda_{n+1-2k_0,n+2-2l}=0.
\eea
Therefore, 
\bea
\lambda_{n+1-2k_0,n+2-2l_0}=\sum_{l=1}^{k_0}\lambda_{n+1-2k_0,n+2-2l}
\eea
is independent of $l_0$. Again, by \eqref{eq:lambda-1}, we have $\lambda_{n+1-2k_0,n+2-2l_0}=0$ for $k_0\ge 2$. We have used the assumption that the characteristic of $F$ is zero.

Similarly, for each fixed $k_0$, the coefficient of $\mathsf{e}_{n+1-2k_0}$ is
\bea
\sum_{l=1}^{\floor{\frac{n-1}{2}}}\sum_{\begin{array}{cc}k=l\\k\ne k_0\end{array}}^{\floor{\frac{n-1}{2}}}\lambda_{n+1-2k,n+2-2l}\sum_{\begin{array}{cc}1\le i_1<\cdots<i_{\floor{\frac{n}{2}}}\le n\\i_m\ne n+1-2k,n+2-2l,n+1-2k_0,\forall m\end{array}}\frac{1}{\prod\limits_{m=1}^{\floor{\frac{n}{2}}}(\zeta_{n+1-2k_0}-\zeta_{{i_m}})}.
\eea
For each fixed $l_0$, the coefficient of the term corresponding to $\{i_1,\dots,i_{\floor{\frac{n}{2}}}\}=\{1,2,\dots,n\}\setminus(\{n+1-2k\mid k=1,2,\dots,\floor{\frac{n-1}{2}}\}\cup\{n+2-2l_0\})$ in the coefficient of $\mathsf{e}_{n+1-2k_0}$ is 
\bea\label{eq:lambda-2}
\sum_{\begin{array}{cc}k=l_0\\k\ne k_0\end{array}}^{\floor{\frac{n-1}{2}}}\lambda_{n+1-2k,n+2-2l_0}=0.
\eea
Therefore, 
\bea
\lambda_{n+1-2k_0,n+2-2l_0}=\sum_{k=l_0}^{\floor{\frac{n-1}{2}}}\lambda_{n+1-2k,n+2-2l_0}
\eea
is independent of $k_0$. Again, by \eqref{eq:lambda-2}, we have $\lambda_{n+1-2k_0,n+2-2l_0}=0$ for $l_0\le \floor{\frac{n-1}{2}}-1$. We have used the assumption that the characteristic of $F$ is zero.

We combine these results and conclude that $\lambda_{n+1-2k,n+2-2l}=0$ for $1\le l\le k\le \floor{\frac{n-1}{2}}$.
\end{proof}

 By Proposition \ref{f-polygon}, the vectors $f^{(n)}_{1\cdots\hat{i}\cdots\hat{j}\cdots n}$ associated to $(n-3)$-simplices in the initial triangulation are linearly independent, so we conclude that the products of matrices on both sides of the $n$-gon equation are equal. Hence we have proved Theorem \ref{main}.
 
 \section{Example: solving the pentagon equation}\label{sec:pentagon}
 
\subsection{Construction of vectors associated to triangles and matrices associated to Pachner 2-2 moves}\label{sec:Pachner-2-2}
In this subsection, we construct vectors associated to triangles and matrices associated to Pachner 2-2 moves. This construction is inspired by \cite{KS13}.

Given 3 pairwise distinct variables $\zeta_i$ ($i=1,2,3$) and a triangle 123, we can construct a vector $(f_1,f_2,f_3)$ such that 
\bea
f_1+f_2+f_3=0\text{ and }f_1\zeta_1+f_2\zeta_2+f_3\zeta_3=0.
\eea
For example,
\bea
(f_1,f_2,f_3)=(\frac{1}{(\zeta_1-\zeta_2)(\zeta_1-\zeta_3)},\frac{1}{(\zeta_2-\zeta_1)(\zeta_2-\zeta_3)},\frac{1}{(\zeta_3-\zeta_1)(\zeta_3-\zeta_2)}).
\eea

For the quadrilateral 1234, the Pachner 2-2 move replaces the triangles 123 and 134 with the triangles 124 and 234.
We will construct a $2\times2$ matrix $P_{24,13}$ for this Pachner move such that
\bea
&&P_{24,13}\left(\begin{array}{cccc}
f_1(123)&f_2(123)&f_3(123)&0\\
f_1(134)&0&f_2(134)&f_3(134)
\end{array}\right)\\
&=&\left(\begin{array}{cccc}
f_1(124)&f_2(124)&0&f_3(124)\\
0&f_1(234)&f_2(234)&f_3(234)
\end{array}\right).\nn
\eea

We solve this matrix equation and get
\bea\label{eq:P-2-2}
P_{24,13}=\left(\begin{array}{cc}
\frac{\zeta_2-\zeta_3}{\zeta_2-\zeta_4}&\frac{\zeta_3-\zeta_4}{\zeta_2-\zeta_4}\\
\frac{\zeta_2-\zeta_1}{\zeta_2-\zeta_4}&\frac{\zeta_1-\zeta_4}{\zeta_2-\zeta_4}\end{array}\right).
\eea
Note that the matrix $P_{24,13}$ is invertible and the sum of the entries of each row of the matrix $P_{24,13}$ is 1.
This gives a solution to the pentagon equation which is essentially the one in \cite{ManturovWanMay2023} after taking the transpose.

\subsection{A solution to the pentagon equation}
In this subsection, we check that \eqref{eq:P-2-2} gives a solution to the pentagon equation.

The pentagon equation is 
\bea\label{eq:pentagon}
d_{35,14}d_{25,13}=d_{24,13}d_{25,14}d_{35,24}
\eea
where $d_{kl,ij}$ ($1\le i<j\le 5$, $1\le k<l\le 5$, $\{i,j\}\cap\{k,l\}=\emptyset$) corresponds to the Pachner 2-2 move in a quadrilateral $ikjl$ that replaces the triangles $ijk$, $ijl$ with the triangles $ikl$, $jkl$ (the order of vertices can be rearranged from small to large).

We will construct matrices $a_{kl,ij}$ for each Pachner 2-3 move satisfying the pentagon equation \eqref{eq:pentagon}.

For the quadrilateral 1234, the Pachner 2-2 move replaces the triangles 123 and 134 with the triangles 124 and 234.
We can construct a $2\times2$ matrix $P_{24,13}$ \eqref{eq:P-2-2} for this Pachner move as we did in \Sec{sec:Pachner-2-2}.

We can extend this matrix by adding a row and a column so that the entry is 1 at the intersection and 0 elsewhere corresponding to each fixed tetrahedron in the Pachner 2-2 moves.
The extended matrix is denoted as $a_{24,13}$.

\begin{figure}[h]
\centering\includegraphics[width = 0.9\textwidth]{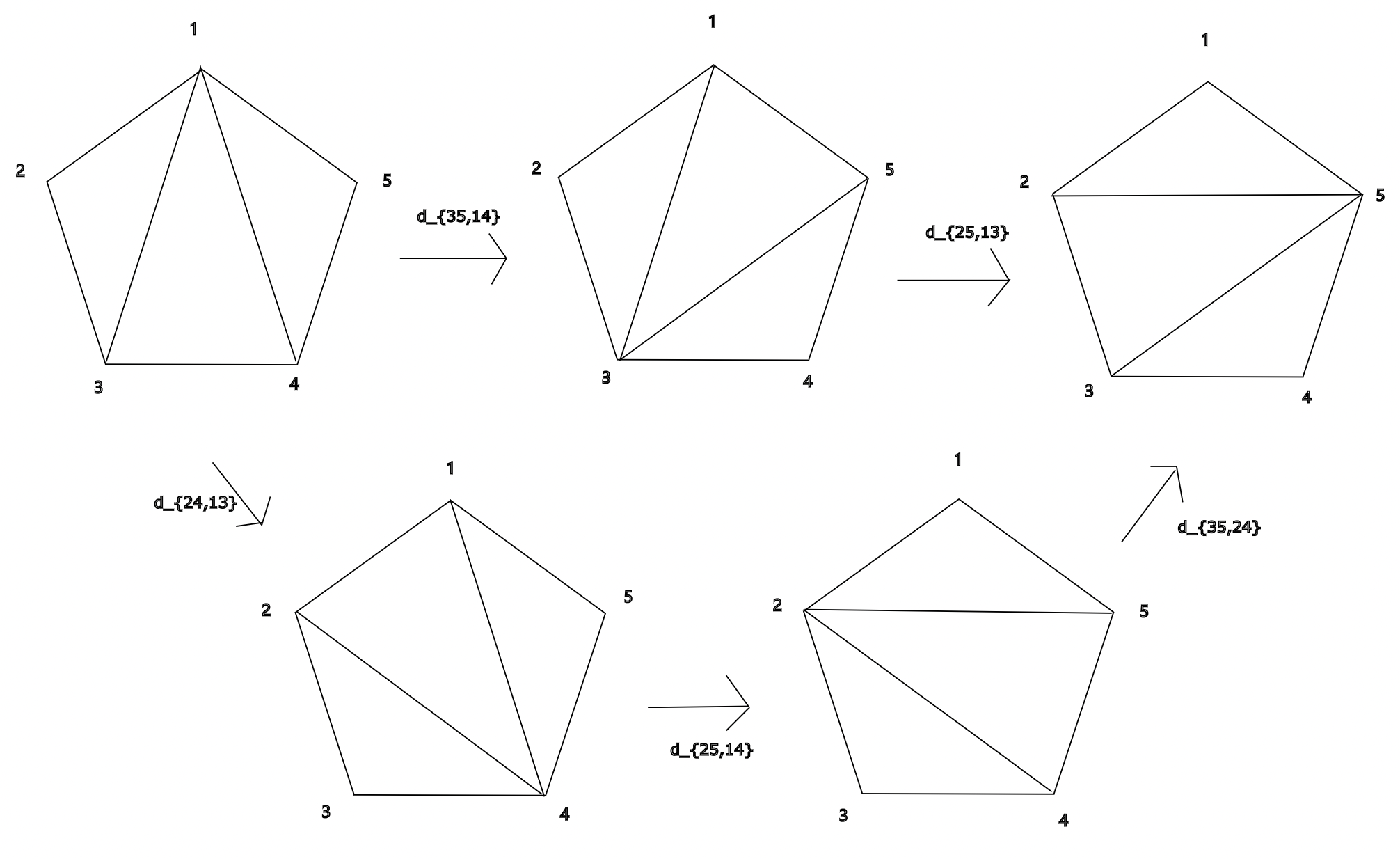}
\caption{Triangulation graph for the pentagon.}\label{fig:pentagon}
\end{figure}

The pentagon equation \eqref{eq:pentagon} involves 5 Pachner 2-2 moves: (1)$\to$(2)$\to$(5) and (1)$\to$(3)$\to$(4)$\to$(5) where (1)-(5) are triangulations of a pentagon as follows (see \Fig{fig:pentagon}):
\begin{enumerate}
\item
123, 134, 145,
\item
123, 135, 345, 
\item
124, 234, 145,
\item
125, 234, 245,
\item
125, 235, 345.
\end{enumerate}
We can check that the matrices $a_{kl,ij}$ for these Pachner 2-2 moves satisfy the pentagon equation \eqref{eq:pentagon}, namely
\bea
&&\left(\begin{array}{ccc}\frac{\zeta_3-\zeta_2}{\zeta_5-\zeta_2}&\frac{\zeta_1-\zeta_2}{\zeta_5-\zeta_2}&\\\frac{\zeta_3-\zeta_5}{\zeta_2-\zeta_5}&\frac{\zeta_1-\zeta_5}{\zeta_2-\zeta_5}&\\&&1\end{array}\right)\left(\begin{array}{ccc}1&&\\&\frac{\zeta_4-\zeta_3}{\zeta_5-\zeta_3}&\frac{\zeta_1-\zeta_3}{\zeta_5-\zeta_3}\\&\frac{\zeta_4-\zeta_5}{\zeta_3-\zeta_5}&\frac{\zeta_1-\zeta_5}{\zeta_3-\zeta_5}\end{array}\right)\nn\\
&=&
\left(\begin{array}{ccc}1&&\\&\frac{\zeta_3-\zeta_4}{\zeta_3-\zeta_5}&\frac{\zeta_4-\zeta_5}{\zeta_3-\zeta_5}\\&\frac{\zeta_3-\zeta_2}{\zeta_3-\zeta_5}&\frac{\zeta_2-\zeta_5}{\zeta_3-\zeta_5}\end{array}\right)
\left(\begin{array}{ccc}\frac{\zeta_2-\zeta_4}{\zeta_2-\zeta_5}&&\frac{\zeta_4-\zeta_5}{\zeta_2-\zeta_5}\\&1&\\\frac{\zeta_2-\zeta_1}{\zeta_2-\zeta_5}&&\frac{\zeta_1-\zeta_5}{\zeta_2-\zeta_5}\end{array}\right)
\left(\begin{array}{ccc}\frac{\zeta_2-\zeta_3}{\zeta_2-\zeta_4}&\frac{\zeta_3-\zeta_4}{\zeta_2-\zeta_4}&\\\frac{\zeta_2-\zeta_1}{\zeta_2-\zeta_4}&\frac{\zeta_1-\zeta_4}{\zeta_2-\zeta_4}&\\&&1\end{array}\right).
\eea
Note that the matrices act from the left, so the order in the products is reversed in comparison to \eqref{eq:pentagon}.

\section{Example: solving the hexagon equation}\label{sec:hexagon}
\subsection{Construction of vectors associated to tetrahedra and matrices associated to Pachner 2-3 moves}\label{sec:Pachner}
In this subsection, we construct vectors associated to tetrahedra and matrices associated to Pachner 2-3 moves. This construction is inspired by \cite{KS13}.

Given 4 pairwise distinct variables $\zeta_i$ ($i=1,2,3,4$) and a tetrahedron 1234, we can construct a vector $(f_1,f_2,f_3,f_4)$ such that 
\bea
f_1\zeta_1^m+f_2\zeta_2^m+f_3\zeta_3^m+f_4\zeta_4^m=0\text{ for }m=0,1,2.
\eea
For example,
\bea
(f_1,f_2,f_3,f_4)&=&(\frac{1}{(\zeta_1-\zeta_2)(\zeta_1-\zeta_3)(\zeta_1-\zeta_4)},\frac{1}{(\zeta_2-\zeta_1)(\zeta_2-\zeta_3)(\zeta_2-\zeta_4)},\\
&&\frac{1}{(\zeta_3-\zeta_1)(\zeta_3-\zeta_2)(\zeta_3-\zeta_4)},\frac{1}{(\zeta_4-\zeta_1)(\zeta_4-\zeta_2)(\zeta_4-\zeta_3)}).\nn
\eea

For the polyhedron 12345, the Pachner 2-3 move (see \Fig{fig:Pachner}) replaces the tetrahedra 1245 and 2345 with the tetrahedra 1234, 1235, and 1345.
We will construct a $3\times2$ matrix $P_{13,245}$ for this Pachner move such that
\bea
&&P_{13,245}\left(\begin{array}{ccccc}
f_1(1245)&f_2(1245)&0&f_3(1245)&f_4(1245)\\
0&f_1(2345)&f_2(2345)&f_3(2345)&f_4(2345)
\end{array}\right)\\
&=&\left(\begin{array}{ccccc}f_1(1234)&f_2(1234)&f_3(1234)&f_4(1234)&0\\
f_1(1235)&f_2(1235)&f_3(1235)&0&f_4(1235)\\
f_1(1345)&0&f_2(1345)&f_3(1345)&f_4(1345)
\end{array}\right).\nn
\eea

We solve this matrix equation and get
\bea\label{eq:P}
P_{13,245}=\left(\begin{array}{ccc}\frac{\zeta_1-\zeta_5}{\zeta_1-\zeta_3}&\frac{\zeta_5-\zeta_3}{\zeta_1-\zeta_3}\\
\frac{\zeta_1-\zeta_4}{\zeta_1-\zeta_3}&\frac{\zeta_4-\zeta_3}{\zeta_1-\zeta_3}\\
\frac{\zeta_1-\zeta_2}{\zeta_1-\zeta_3}&\frac{\zeta_2-\zeta_3}{\zeta_1-\zeta_3}\end{array}\right).
\eea
Note that the sum of the entries of each row of the matrix $P_{13,245}$ is 1.

\subsection{A solution to the hexagon equation}
In this subsection, we check that \eqref{eq:P} gives a solution to the hexagon equation.

The hexagon equation is
\bea\label{eq:hexagon}
d_{(46,125)}d_{(36,124)}d_{(35,246)}=d_{(35,124)}d_{(36,125)}d_{(46,135)}
\eea
where $d_{(ij,klm)}$ ($1\le i<j\le 6$, $1\le k<l<m\le 6$, $\{i,j\}\cap\{k,l,m\}=\emptyset$) corresponds to the Pachner 2-3 move in a polyhedron $ijklm$ with 6 faces where each face is a triangle, $ij$ is the vertical axis, and $klm$ is the horizontal triangle.
It replaces the tetrahedra $iklm$, $jklm$ with the tetrahedra $ijkl$, $ijkm$, $ijlm$ (the order of vertices can be rearranged from small to large).

We will construct matrices $a_{ij,klm}$ for each Pachner 2-3 move satisfying the hexagon equation \eqref{eq:hexagon}.

For the polyhedron 12345, the Pachner 2-3 move (see \Fig{fig:Pachner}) replaces the tetrahedra 1245 and 2345 with the tetrahedra 1234, 1235, and 1345.
We can construct a $3\times2$ matrix $P_{13,245}$ \eqref{eq:P} for this Pachner move as we did in \Sec{sec:Pachner}.

We can extend this matrix by adding a row and a column so that the entry is 1 at the intersection and 0 elsewhere corresponding to each fixed tetrahedron in the Pachner 2-3 moves.
The extended matrix is denoted as $a_{13,245}$.

\begin{figure}[h]
\centering\includegraphics[width = 0.9\textwidth]{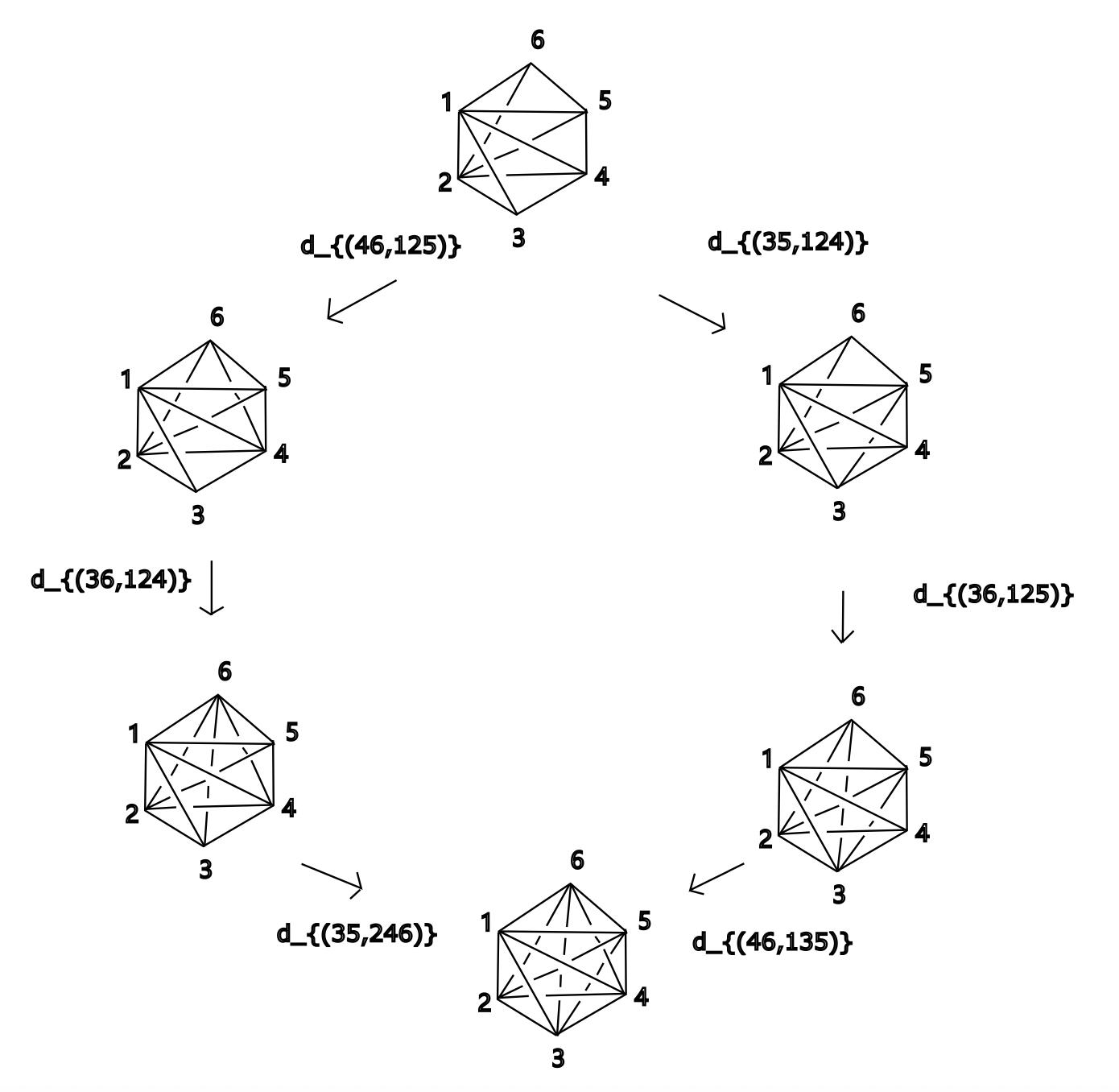}
\caption{Triangulation graph for the shifted octahedron.}\label{fig:octahedron}
\end{figure}

The hexagon equation \eqref{eq:hexagon} involves 6 Pachner 2-3 moves: (1)$\to$(2)$\to$(4)$\to$(6) and (1)$\to$(3)$\to$(5)$\to$(6) where (1)-(6) are triangulations of a shifted octahedron as follows (see \Fig{fig:octahedron}):
\begin{enumerate}
\item
1234, 1245, 1256,
\item
1234, 1246, 1456, 2456,
\item
1235, 1256, 1345, 2345,
\item
1236, 1346, 1456, 2346, 2456,
\item
1236, 1345, 1356, 2345, 2356,
\item
1236, 1346, 1456, 2345, 2356, 3456.
\end{enumerate}
We can check that the matrices $a_{ij,klm}$ for these Pachner 2-3 moves satisfy the hexagon equation \eqref{eq:hexagon}, namely
\bea
&&\left(\begin{array}{cccccc}1&0&0&0&0\\0&1&0&0&0\\0&0&1&0&0\\0&0&0&\frac{\zeta_3-\zeta_6}{\zeta_3-\zeta_5}&\frac{\zeta_6-\zeta_5}{\zeta_3-\zeta_5}\\
0&0&0&\frac{\zeta_3-\zeta_4}{\zeta_3-\zeta_5}&\frac{\zeta_4-\zeta_5}{\zeta_3-\zeta_5}\\
0&0&0&\frac{\zeta_3-\zeta_2}{\zeta_3-\zeta_5}&\frac{\zeta_2-\zeta_5}{\zeta_3-\zeta_5}\end{array}\right)
\left(\begin{array}{ccccc}\frac{\zeta_3-\zeta_4}{\zeta_3-\zeta_6}&\frac{\zeta_4-\zeta_6}{\zeta_3-\zeta_6}&0&0\\
\frac{\zeta_3-\zeta_2}{\zeta_3-\zeta_6}&\frac{\zeta_2-\zeta_6}{\zeta_3-\zeta_6}&0&0\\
0&0&1&0\\
\frac{\zeta_3-\zeta_1}{\zeta_3-\zeta_6}&\frac{\zeta_1-\zeta_6}{\zeta_3-\zeta_6}&0&0\\
0&0&0&1\end{array}\right)
\left(\begin{array}{cccc}1&0&0\\0&\frac{\zeta_4-\zeta_5}{\zeta_4-\zeta_6}&\frac{\zeta_5-\zeta_6}{\zeta_4-\zeta_6}\\
0&\frac{\zeta_4-\zeta_2}{\zeta_4-\zeta_6}&\frac{\zeta_2-\zeta_6}{\zeta_4-\zeta_6}\\
0&\frac{\zeta_4-\zeta_1}{\zeta_4-\zeta_6}&\frac{\zeta_1-\zeta_6}{\zeta_4-\zeta_6}\end{array}\right)\\
&=&\left(\begin{array}{cccccc}1&0&0&0&0\\0&\frac{\zeta_4-\zeta_5}{\zeta_4-\zeta_6}&\frac{\zeta_5-\zeta_6}{\zeta_4-\zeta_6}&0&0\\
0&\frac{\zeta_4-\zeta_3}{\zeta_4-\zeta_6}&\frac{\zeta_3-\zeta_6}{\zeta_4-\zeta_6}&0&0\\
0&0&0&1&0\\
0&0&0&0&1\\
0&\frac{\zeta_4-\zeta_1}{\zeta_4-\zeta_6}&\frac{\zeta_1-\zeta_6}{\zeta_4-\zeta_6}&0&0\end{array}\right)
\left(\begin{array}{ccccc}\frac{\zeta_3-\zeta_5}{\zeta_3-\zeta_6}&\frac{\zeta_5-\zeta_6}{\zeta_3-\zeta_6}&0&0\\
0&0&1&0\\
\frac{\zeta_3-\zeta_2}{\zeta_3-\zeta_6}&\frac{\zeta_2-\zeta_6}{\zeta_3-\zeta_6}&0&0\\
0&0&0&1\\
\frac{\zeta_3-\zeta_1}{\zeta_3-\zeta_6}&\frac{\zeta_1-\zeta_6}{\zeta_3-\zeta_6}&0&0\end{array}\right)
\left(\begin{array}{cccc}\frac{\zeta_3-\zeta_4}{\zeta_3-\zeta_5}&\frac{\zeta_4-\zeta_5}{\zeta_3-\zeta_5}&0\\
0&0&1\\
\frac{\zeta_3-\zeta_2}{\zeta_3-\zeta_5}&\frac{\zeta_2-\zeta_5}{\zeta_3-\zeta_5}&0\\
\frac{\zeta_3-\zeta_1}{\zeta_3-\zeta_5}&\frac{\zeta_1-\zeta_5}{\zeta_3-\zeta_5}&0\end{array}\right).\nn
\eea
Note that the matrices act from the left, so the order in the products is reversed in comparison to \eqref{eq:hexagon}.
A different solution to the hexagon equation was given in \cite{KS17}.

\section{Example: solving the heptagon equation}\label{sec:heptagon}
\subsection{Construction of vectors associated to 4-simplices and matrices associated to Pachner 3-3 moves}
In this subsection, we construct vectors associated to 4-simplices and matrices associated to Pachner 3-3 moves. This construction is inspired by \cite{KS13}.

Given 5 pairwise distinct variables $\zeta_i$ ($i=1,2,3,4$) and a 4-simplex 12345, we can construct a vector $(f_1,f_2,f_3,f_4,f_5)$ such that 
\bea
f_1\zeta_1^m+f_2\zeta_2^m+f_3\zeta_3^m+f_4\zeta_4^m+f_5\zeta_5^m=0\text{ for }m=0,1,2.
\eea
For example,
\bea
(f_1,f_2,f_3,f_4,f_5)=(f(1,2,3,4,5), f(2,1,3,4,5), f(3,1,2,4,5), f(4,1,2,3,5), f(5,1,2,3,4))
\eea
where
\bea
f(i,j,k,l,m) := \frac{4\zeta_i-\zeta_j-\zeta_k-\zeta_l-\zeta_m}{(\zeta_i-\zeta_j)(\zeta_i-\zeta_k)(\zeta_i-\zeta_l)(\zeta_i-\zeta_m)}.
\eea
Let $m$ be the matrix
\bea
\left(\begin{array}{cccccc} 0 & f(2,3,4,5,6) & f(3,2,4,5,6) & f(4,2,3,5,6) & f(5,2,3,4,6) & f(6,2,3,4,5) \\
                     f(1,3,4,5,6)& 0& f(3,1,4,5,6)& f(4,1,3,5,6)& f(5,1,3,4,6)& f(6,1,3,4,5) \\
                     f(1,2,4,5,6)& f(2,1,4,5,6)& 0& f(4,1,2,5,6)& f(5,1,2,4,6)& f(6,1,2,4,5) \\
                     f(1,2,3,5,6)& f(2,1,3,5,6)& f(3,1,2,5,6)& 0& f(5,1,2,3,6)& f(6,1,2,3,5) \\
                     f(1,2,3,4,6)& f(2,1,3,4,6)& f(3,1,2,4,6)& f(4,1,2,3,6)& 0& f(6,1,2,3,4) \\
                     f(1,2,3,4,5)& f(2,1,3,4,5)& f(3,1,2,4,5)& f(4,1,2,3,5)& f(5,1,2,3,4)& 0 
                    \end{array}\right),
\eea
then the rank of $m$ is 3.

For the polyhedron 123456, the Pachner 3-3 move replaces 3 of 6 4-simplices with the other 3, \ie the 4-simplices 12345, 12356, and 13456 with the 4-simplices 12346, 12456, and 23456.
We will construct a $3\times3$ matrix $P_{246,135}$ for this Pachner move such that
\bea
&&P_{246,135}\cdot\\
&&\left(\begin{array}{cccccc} 
f(1,2,3,4,5)& f(2,1,3,4,5)& f(3,1,2,4,5)& f(4,1,2,3,5)& f(5,1,2,3,4)& 0 \\
f(1,2,3,5,6)& f(2,1,3,5,6)& f(3,1,2,5,6)& 0& f(5,1,2,3,6)& f(6,1,2,3,5) \\
f(1,3,4,5,6)& 0& f(3,1,4,5,6)& f(4,1,3,5,6)& f(5,1,3,4,6)& f(6,1,3,4,5)                      
\end{array}\right)\nn\\
&=&\left(\begin{array}{cccccc} 
                     f(1,2,3,4,6)& f(2,1,3,4,6)& f(3,1,2,4,6)& f(4,1,2,3,6)& 0& f(6,1,2,3,4) \\
                     f(1,2,4,5,6)& f(2,1,4,5,6)& 0& f(4,1,2,5,6)& f(5,1,2,4,6)& f(6,1,2,4,5) \\
                     0 & f(2,3,4,5,6) & f(3,2,4,5,6) & f(4,2,3,5,6) & f(5,2,3,4,6) & f(6,2,3,4,5) 
\end{array}\right).\nn
\eea

We solve this matrix equation and get
\bea
P_{246,135}
&=&\left(\begin{array}{ccc}\frac{V(2,4,5)}{V(2,4,6)}&\frac{V(2,5,6)}{V(2,4,6)}&\frac{V(5,4,6)}{V(2,4,6)}\\
\frac{V(2,4,3)}{V(2,4,6)}&\frac{V(2,3,6)}{V(2,4,6)}&\frac{V(3,4,6)}{V(2,4,6)}\\
\frac{V(2,4,1)}{V(2,4,6)}&\frac{V(2,1,6)}{V(2,4,6)}&\frac{V(1,4,6)}{V(2,4,6)}
\end{array}\right)
\eea
where $V(i,j,k)=\begin{vmatrix}1&1&1\\\zeta_i&\zeta_j&\zeta_k\\\zeta_i^2&\zeta_j^2&\zeta_k^2\end{vmatrix}=(\zeta_i-\zeta_j)(\zeta_i-\zeta_k)(\zeta_j-\zeta_k)$
is the Vandermonde determinant. Note that the matrix $P_{246,135}$ is invertible and the sum of the entries of each row of the matrix $P_{246,135}$ is 1 because $\begin{vmatrix}1&1&1&1\\1&1&1&1\\\zeta_1&\zeta_2&\zeta_3&\zeta_4\\\zeta_1^2&\zeta_2^2&\zeta_3^2&\zeta_4^2\end{vmatrix}=V(2,3,4)-V(1,3,4)+V(1,2,4)-V(1,2,3)=0$. This gives a solution to the heptagon equation which is exactly a partial case of the one in \cite{DK21}.
A similar solution to the heptagon equation was given in \cite{K22a} and similar solutions to general odd-gon equations were given in \cite{K22b}.

\end{document}